\documentclass{gAPA2e}
\allowdisplaybreaks
\begin{document}
\doi{10.1080/0003681YYxxxxxxxx}
 \issn{1563-504X}
\issnp{0003-6811}
\jvol{00} \jnum{00} \jyear{2008} \jmonth{January}

\markboth{A. Ya. Olenko \& T. K. Pog\'any}{UNIVERSAL TRUNCATION ERROR UPPER BOUNDS IN SAMPLING RESTORATION}

\title{UNIVERSAL TRUNCATION ERROR UPPER BOUNDS IN IRREGULAR SAMPLING RESTORATION}

\author{Andriy Ya. Olenko $^{\rm a}$, Tibor K. Pog\'any $^{\rm b,}$$^\star$\thanks{$^\star$ Corresponding author.
E--mail: \texttt{poganj@pfri.hr}\vspace{6pt}}\\
$^{\rm a}$ {\it Department of Mathematics and Statistics, LaTrobe University,  Victoria, 3086, Australia;
{\tt a.olenko@latrobe.edu.au}\\
$^{\rm b}$ Faculty of Maritime Studies, University of Rijeka, HR-51000 Rijeka, Studentska 2, Croatia;
{\tt poganj@pfri.hr}}. \\\vspace{6mm}
\received{ This is an Author's Accepted Manuscript of an article published in the
Applicable Analysis, Vol.~90, No. 3-4. (2011), 595--608.  [copyright Taylor \& Francis], available online at: http://www.tandfonline.com/ [DOI:10.1080/00036810903437754]}
}

\maketitle

\begin{abstract}
Universal (pointwise uniform and time shifted) truncation error upper bounds are presented in
Whittaker--Kotel'nikov--Shannon (WKS) sampling restoration sum for
Bernstein  function class $B_{\pi,d}^q\,,\ q \ge 1,$ $d\in \mathbb N\,,$
when the sampled functions decay rate is unknown. The case of
multidimensional irregular sampling is discussed.
\bigskip

\begin{classcode} {Primary 94A20, 26D15; Secondary 30D15, 41A05}\end{classcode}
\bigskip
\end{abstract}

\section{Introduction}
The classical WKS sampling theorem has been extended to the case of nonuniform sampling by numerous authors.
For detailed information on the theory and its numerous applications, we refer to the book \cite{mar}.

Most known irregular sampling results deal with Paley--Wiener functions which have $L^2(\mathbb R)$ restrictions
on the real line. It seems that the best known nonuniform WKS
sampling results for entire functions in $L^p$--spaces were given in \cite{hin, hin1}. There are no
multidimensional $L^p$--WKS sampling theorems with precise truncation error estimates in open literature.
However, explicit truncation error upper bounds in multidimensional WKS reconstructions are of
great importance in signal and image processing.

Alternative reconstruction approaches for irregular sampling problems were developed in \cite{fei1,fei2,fei3}. However,
due to long traditions the WKS type reconstructions still play key role in applied signal and image processing.
New numerical methods for quick sinc--series summations (see, for example, \cite{gre}) let more efficient
usage of WKS formulae than before. On the other hand WKS type results
are important not only because of signal processing applications. WKS theorems are equivalent to many key results
in mathematics, see, e.g. \cite{hig2}. Therefore they are also valuable for theoretical studies.
It explains why in spite of substantial progress in modern approximation methods WKS type reconstructions
are still of great importance and vast amount of new refine results are published regularly.

In this paper we use approaches and methods developed in \cite{FLS,OP,OP1,OP2,uni_1} to
investigate multidimensional irregular sampling in $L^p$-spaces.

Let $\mathsf X$ be a normed space and assume that the structure of
$\mathsf X$ admits the sampling approximation procedure
  \begin{equation}\label{1}
      f(\mathbf{x}) = \sum_{\mathbf{n}\in \mathbb Z^d} f(t_\mathbf{n})S(\mathbf{x},t_\mathbf{n}) \qquad
      \big(f\in \mathsf X\big)
  \end{equation}
where $\mathbf{x} \in \mathbb R^d$; $\mathfrak T :=
\{t_\mathbf{n}\}_{\mathbf{n} \in \mathbb Z^d} \subset \mathbb R^d$
and $S(\cdot,\cdot)$ are the sampling set and the sampling function
respectively. This formula is one of the basic tools in signal
processing.

In direct numerical implementations we consider the truncated
variant of (\ref{1}), reads as follows:
   \[ Y_{\mathfrak J}(f;\mathbf{x}) = \sum_{\mathbf{n}\in \mathfrak J} f(t_\mathbf{n})S(\mathbf{x},t_\mathbf{n}),
                                      \qquad \big(\mathfrak J \subset \mathbb Z^d\big), \]
where the index set $\mathfrak J$ is finite by application reasons.
Namely, recovering a continuous signal from discrete samples or
assessing the information lost in the sampling process are the
fundamental problems in sampling and interpolation theory.

The usual procedure is to estimate the truncation error
   \begin{equation} \label{3}
      \| T_{\mathfrak J}(f;\mathbf{x})\| :=  \|f(\mathbf{x}) -
          Y_{\mathfrak J}(f;\mathbf{x})\| \le \varphi_{\mathfrak J}(f;\mathbf{x}),
   \end{equation}
where $\| \cdot\|$ could stay for different norms on different sides
of this relation and $\varphi_{\mathfrak J}$ denotes the truncation
error upper bound. Simple truncation error upper bounds are the main
tools in numerical implementations, when they do not contain
infinite products and/or unknown function values. However, the
suitably sharp truncation error upper bound enables pointwise,
almost sure, uniform {\em etc.} type convergence of the approximating
sequence $Y_{\mathfrak J}(f;\mathbf{x})$ to initial $f \in \mathsf X$. The main aim of this
paper is for $f\in L^p$ to discuss
$\varphi_{\mathfrak J}(f; \mathbf x)$ by getting pointwise upper
bounds valid for all $\mathbf x$ on the whole range of the signal
domain without strong assumptions on the decay rate of $f,$ see \cite{OP2}. This kind of upper bounds
we call {\em universal}.

The most frequent rearrangement of \eqref{3} in the literature is of the form
   \begin{equation} \label{4}
      \| T_{\mathfrak J}(f;\mathbf{x})\|_\infty \le
         \Big( \sum_{\mathbf n \in \mathbb Z^d\setminus \mathfrak J}|S(\mathbf x,t_{\mathbf n})|^p\Big)^{1/p}
         \Big( \sum_{\mathbf n \in \mathbb Z^d\setminus \mathfrak J}|f(t_{\mathbf n})|^q\Big)^{1/q}
         =: A_p\,B_q\, ,
   \end{equation}
$p, q$ being a conjugated H\"older pair, i.e. $1/p+1/q=1$.

To make the approximant $Y_{\mathfrak J}(f;\mathbf{x})$ more precise
let us assume $\mathfrak J = \mathfrak J_\mathbf x$, that is, that the sampling index set
$\mathfrak J_{\mathbf x}$ depends on the location of time parameter $\mathbf x$ with
respect to the behaviour of $f$ in estimating $B_q,$ see  \cite{FLS, OP1, Yen}. That means
$T_{\mathfrak J}(f;\mathbf{x})$ is {\em a fortiori} time shifted and possesses time adapted sampling size.

To obtain a class of truncation error upper bounds when the decay
rate of the initial signal function is not known one operates with
the straightforward $B_q \le C_{f,\mathfrak T}\,\|f\|_q$ where
$C_{f, \mathfrak T}$ is suitable absolute constant. Thus \eqref{4} becomes
   \[ \| T_{\mathfrak J}(f;\mathbf{x})\|_\infty \le C_{f,\mathfrak T}\,A_p\, \|f\|_q. \]
We are interested in estimates for $A_p$ such that vanish with
$|\mathfrak J_{\mathbf x}| \to \infty .$ Therefore, the obtained upper bounds really will be universal for wide
classes of $f(\mathbf x)$ and $\mathfrak T$.

\section{Multidimensional Plancherel--P\'olya inequality.}
Denote $\| \cdot \|_p$ the $L_p$-norm in finite case (while $\| \cdot\|_\infty \equiv {\rm sup\,\, vrai}|\cdot |)$
and $L^p(\mathbb R)$ the set of all functions whose restrictions to $\mathbb R$ have finite $L_p$-norm.
Assume $f \in L^r(\mathbb R),\, r>0$ be of exponential type $\sigma>0$ and let $\{ t_n\}_{n\in \mathbb Z}$
be a separated real sequence, i.e. such that $\inf_{n\neq m}|t_n-t_m|\ge \delta>0\,.$ Then it holds true
\cite[Eq.~(76)]{PPII}
   \begin{equation} \label{PP}
      \sum_{n \in \mathbb Z} |f(t_n)|^r \le B \|f\|_r^r \, ,
   \end{equation}
where
   \[ \displaystyle B = \frac{8(e^{r\,\delta \sigma /2} -1)}{ r\, \pi \sigma\delta^2 } \, .\]
The display (\ref{PP}) exposes the celebrated Plancherel--P\'olya
inequality, see \cite{Boas, PP, PPII}. It could be mentioned that Boas \cite{Boas0} has been established
another estimate in one--dimensional case under different assumptions on $\mathfrak T$, and
recently Lindner publishes an estimate in onedimensional case when $r = 2,$ \cite{lin}.

Here, we will give the multidimensional analog of the Plancherel--P\'olya inequality. Hereinafter $B_{\boldsymbol
\sigma,d}^r,\, r>0$ denotes the Bernstein class \cite{nik} of
$d$--variable entire functions of exponential type at most $\boldsymbol \sigma = (\sigma_1, \cdots, \sigma_d)$
coordinatewise whose restriction to $\mathbb R^d$ is in $L^r(\mathbb R^d)$.
\begin{theorem}\label{th1} Let $\mathfrak T = \{ t_{\mathbf n}\}_{\mathbf n \in \mathbb Z^d},\,
t_{\mathbf n}=(t_{n_1},...t_{n_d})$ be real  separated sequence, i.e.
$\inf_{n_\ell \neq m_\ell}|t_{n_\ell}-t_{m_\ell}|\ge \delta_\ell>0,\, \ell = \overline{1,d}$. Let
$f \in B_{\boldsymbol \sigma,d}^r,\, r\ge 1$. Then
   \[ \sum_{\mathbf n \in \mathbb Z^d} |f(t_{\mathbf n})|^r \le \mathfrak B_{d,r} \|f\|_r^r \, ,\]
where
   \[ \displaystyle \mathfrak B_{d,r} = \Big( \frac{8}{r \pi}\Big)^d \prod_{\ell=1}^d
                    \frac{e^{r\,\delta_\ell \sigma_\ell /2} -1}{\sigma_\ell \delta_\ell^2} \, . \]
\end{theorem}
\begin{proof} Take $d=2$; the proof will be identical in the case $d>2$.

From the assumption $f \in B_{\boldsymbol \sigma,d}^r$ and the inequality of different dimensions
\cite[\S 3.4.2]{nik} which holds for $r\ge 1$ follow that $f(\cdot, x_2)\in B_{\sigma_1,1}^r$ and
$f(x_1, \cdot)\in B_{\sigma_2,1}^r.$ Therefore we can apply (\ref{PP}) coordinatewise.
Because $\{t_{n_\ell}\}_{n_\ell\in \mathbb Z}$ are separated with $\delta_\ell,\, \ell=1,2$ there holds
   \begin{align}
      \sum_{\mathbf n \in \mathbb Z^2} |f(t_{\mathbf n})|^r &=
              \sum_{n_1 \in \mathbb Z} \Bigg( \sum_{n_2 \in \mathbb Z}| f(t_{n_1},t_{n_2})|^r \Bigg)
          \le \frac{8(e^{r\,\delta_2 \sigma_2 /2} -1)}{ r\, \pi\sigma_2 \delta_2^2 }
              \sum_{n_1 \in \mathbb Z} \int_{\mathbb R}|f(t_{n_1},x_2)|^r\, dx_2\, \nonumber \\  \label{delta2}
         &= \frac{8(e^{r\,\delta_2 \sigma_2 /2} -1)}{ r\, \pi\sigma_2 \delta_2^2 }
              \int_{\mathbb R} \Bigg( \sum_{n_1 \in \mathbb Z} |f(t_{n_1},x_2)|^r\Bigg) dx_2 \, .
   \end{align}
Second subsequent application of (\ref{PP}) to (\ref{delta2}) yields
   \begin{equation} \label{delta1}
      \sum_{\mathbf n \in \mathbb Z^2} |f(t_{\mathbf n})|^r \le
            \frac{64(e^{r\,\delta_1 \sigma_1 /2} -1)(e^{r\,\delta_2 \sigma_2 /2} -1)}
                 {\sigma_1\sigma_2\, r^2\, \pi^2 (\delta_1\delta_2)^2 }
                 \int_{\mathbb R^2} |f(x_1,x_2)|^r dx_1 dx_2 = \mathfrak B_{2,r} \|f\|_r^r\, ,
   \end{equation}
where in (\ref{delta1}) the Fubini's theorem is used. Therefore, the assertion  of Theorem is proved.
\end{proof}

\section{Multidimensional irregular sampling}
The case of WKS sampling restoration by nonequidistantly spaced nodes
$\mathfrak T  = \{t_{\mathbf n} = \mathbf n + h_{\mathbf n} \}$ is the so--called
{\em irregular sampling procedure}. In this case there is a doubt on the reality of
analogue of regular sampling results on sharp estimates for truncation error upper
bounds in uniform sampling (for example, see \cite{uni_1}), since $\mathfrak h := \{h_{\mathbf n}\}$ is
wandering together with $\mathbf x$, and therefore different time values $\mathbf x$
generate windows with substantially different sampling sets. Of course, this results in increasing
magnitude of the truncation error \cite{Pog3}.

Let $\mathbf N:= (N_1, \cdots, N_d)\in \mathbb N^d,\,\mathfrak J_{\mathbf x} := \{ \mathbf n:\,
\bigwedge_{j = 1}^d(|x_j - n_j|\le N_j)\}$. In \cite{FLS, Pog2} it was considered a case when outside
$\mathfrak J_{\mathbf x}$ the time-jitter $\mathfrak h$ vanishes. These results can be generalized
to multidimensional case. Namely, having on mind the infinite product representation of sine--function and replacing
the equidistant nodes having indices inside $\mathfrak J_{x_j}$ with the disturbed $t_{ n_j}$ ones $j=\overline{1,d}$
we get the so--called {\em window canonical product sampling function}
   \begin{equation} \label{19}
      S(\mathbf x, t_{\mathbf n}) = \prod_{j=1}^d \frac{G_{N_j}(x_j,x_j)}{G_{N_j}'(x_j,t_{n_j})(x_j-t_{n_j})}
   \end{equation}
where $G_{N}'(x,t)$ denotes a derivative with respect to $t$, being
   \begin{align} \label{20}
           G_N(x,t) &=  (t-h_0)\, {\rm sinc}(t) \mathop{\prod_{|x-k|\le N}}_{k\neq 0}
                        \left( 1- \frac{h_k }{t-k}\right)\frac{k}{t_k}\, , \\
      {\rm sinc}(t) &:= \begin{cases} \displaystyle
                           \frac{\sin (\pi t)}{\pi t} &\quad \hbox{if}\ \ t \neq 0 \\
                                                    1 &\quad \hbox{if}\ \ t=0
                        \end{cases} \, . \nonumber
   \end{align}
So, when $\mathfrak T$ is separated, one can deduce
   \begin{equation} \label{21}
      f(\mathbf x) = \sum_{\mathbf n \in \mathbb Z^d}f(\mathbf t_{\mathbf n})\prod_{j=1}^d
                     \frac{G_{N_j}(x_j,x_j)}{G_{N_j}'(x_j,t_{n_j})(x_j-t_{n_j})},
   \end{equation}
where the convergence is uniform on bounded subsets of $\mathbb R^d$ under additional suitable conditions
upon $f$ and $\mathfrak T$. We will discuss the conditions for multidimensional case later
(see \cite{hig, hin, mar, Pog3, Seip} for onedimensional case). In this framework the sampling
restoration procedure becomes of Lagrange--Yen type \cite{FLS, mar, Pog3, Yen}.

For the general case of  multidimensional irregular sampling with window canonical product
sampling function $S(\mathbf x, t_{\mathbf n})$ we have time--jittered nodes outside $\mathfrak J_{\mathbf x}$
as well. Nonvanishing time-jitter $\mathfrak h$ outside $\mathfrak J_{\mathbf x}$ leads to
functions $G_N(x,t)$ given by formulae different from (\ref{20}), see \cite{hin}. However, in irregular
sampling applications we would like to approximate $f(\mathbf x)$ using only it's values from sample nodes
$t_\mathbf{n}$ indexed by $\mathfrak J_{\mathbf x}.$ Therefore we can try to use the truncated to
$\mathfrak J_{\boldsymbol x}$ sampling approximation sum
   \[ Y_{\mathfrak J_{\mathbf x}}(f;\mathbf{x})=\sum_{\mathbf n \in \mathfrak
                   J_{\mathbf x}} f(t_\mathbf{n})\prod_{j=1}^d
                   \frac{G_{N_j}(x_j,x_j)}{G_{N_j}'(x_j,t_{n_j})(x_j-t_{n_j})}\]
with $G_{N}(x,t)$ (such that is given by \eqref{20}) even for arbitrary sample nodes outside
$\mathfrak J_{\mathbf x}$. Under such assumptions the {\em truncation error}
   \[ \| T_{\mathbf{N},d}(f;\mathbf x)\|_\infty=\|f(\mathbf x) - Y_{\mathfrak J_{\mathbf x}}(f;\mathbf{x})\|_\infty \]
coincides with truncation error for the case given by \eqref{19}--\eqref{21}. It would seem that in
this framework $T_{\mathbf{N},d}(f;\mathbf x) = Y_{\mathbb Z^d \setminus \mathfrak J_{\mathbf x}}(f;\mathbf{x})$
depends on samples in $\mathfrak T\setminus \{t_\mathbf{n}:\,\mathbf{n}\in \mathfrak J_{\mathbf x}\}$,
but it does not depend on samples in $\{t_\mathbf n \colon \mathbf n \in \mathfrak J_{\mathbf x}\}$, which
are used to build approximation sum $Y_{\mathfrak J_{\mathbf x}}(f;\mathbf{x})$.
However, it is not true, because $T_{\mathbf{N},d}(f;\mathbf x)$ depends on nonvanishing
$\mathfrak h$ in $\mathfrak J_{\mathbf x}$ due to multiplicative form of (\ref{21}).
Therefore, the multidimensional sampling problems are more difficult than the one--dimensional ones.
In this section we propose a way  to obtain universal truncation bounds for multidimensional irregular sampling
restoration procedure.

First of all, we will give the multidimensional sampling theorem for $B_{\boldsymbol \sigma,d}^r$ functional class.

Denote $\mathbf M=(M_1,\cdots,M_d),\, \boldsymbol{\delta}=(\delta_1,\cdots,\delta_d),\,\widetilde{M} :=
\max_{j=\overline{1,d}}M_j.$  Assume that $t_{n_j} = n_j + h_{n_j},\, |h_{n_j}|\le M_j,\,j=\overline{1,d}$;
for all $\mathbf n \in \mathfrak J_{\mathbf x}$.

\begin{theorem}\label{th_ms} Let $f\in B^q_{\boldsymbol \sigma,d},$ $q\ge 1,$ $\sigma_j\le \pi$ for all $j,$
$\mathfrak T = \{ t_{\mathbf n}\}_{\mathbf n \in \mathbb Z^d}$ be real separated sequence with
   \begin{equation} \label{M}
      \widetilde M \le \frac14\, \chi_{\{q=1\}} \quad \textit{and} \quad
      \widetilde M < \frac1{4q}\, \chi_{\{1<q<\infty\}}\, .
   \end{equation}
Then the sampling expansion {\rm (\ref{21})} holds uniformly on each bounded $\mathbf x$--subset of $\Bbb R^d$.
Moreover, the series in {\rm  (\ref{21})} converges absolutely too. Here, in \eqref{M}, $\chi_A$ stands for the
characteristic function of the random event $A$.
\end{theorem}
\begin{proof} First $\mathfrak T$ is obviously separated since
   \[ \delta_j =   \inf_{n_{j}\in \mathbb Z, \,k \in \mathbb N}\big|t_{n_j+k}-t_{n_j}\big|
               \ge 1-2 M_j \ge 1-2\widetilde M > 1-\frac1{2q}>0  \qquad \big(\mathbf n \in
               \mathfrak J_{\mathbf x}\big).\]
Now, by the assumption $f\in B^q_{\boldsymbol \sigma,d}, \ q\ge 1$ and by the inequality of different dimensions
\cite[\S 3.4.2]{nik} we conclude $f(x_1,\dots,x_{j-1},\cdot,x_{j+1},\dots,x_{n})\in B^q_{\sigma_j,1}, j=\overline{1,d}$.
On the other hand we rewrite the function $G_N(x,t)$ (defined by \eqref{20}) into
   \begin{equation} \label{G}
      G_N(x,t) = (t-h_0) \prod_{n=1}^\infty \left(1-\frac{t}{\widetilde{t}_n}\right)
                 \left(1-\frac{t}{\widetilde{t}_{-n}}\right),
   \end{equation}
where
   \[ \widetilde{t}_n = \begin{cases}
                           t_n &\quad \hbox{if}\ \ |x-n|\le N \\
                            n  &\quad \hbox{else.}
                        \end{cases} \, .\]
Because $G_N(x,t)$ has representation (\ref{G}), it is clear that all results of \cite{hin} remain valid if
one uses $h_0 \neq 0$ instead of $h_0=0$, and make use of $G_N(x,t)$ instead of
$G(t)=t\prod_{n=1}^\infty \big(1-t/{t_n}\big)\big(1-t/{t_{-n}}\big)$, considered by Hinsen \cite{hin}.
Therefore, if (\ref{M}) holds, we have
   \begin{equation} \label{S1}
      f(x_1,\dots,x_{n}) = \sum_{n_j \in \mathbb Z}f(x_1,\dots,x_{j-1},t_{n_j},x_{j+1},\dots,x_{n})
             \frac{G_{N_j}(x_j,x_j)}{G_{N_j}'(x_j,t_{n_j})(x_j-t_{n_j})}.
   \end{equation}
We can subsequently apply the same sampling expansion formula (\ref{S1}) with respect to all variables $x_k$ on
the right side of (\ref{S1}), being (\ref{S1}) absolutely convergent for $\widetilde M<1/(4q)$ \cite{hin}.
This procedure results in absolute convergent $d$--tuple series on the right in (\ref{21}).
\end{proof}

\section{Truncation error upper bounds}
The belonging truncation error upper bound result reads as follows.
\begin{theorem}\label{th_ir} Let $\widetilde{M}$ satisfy {\rm \eqref{M}}, $f\in B^q_{\boldsymbol \sigma,d},\,
q\ge 1,\, \sigma_j\le \pi$ for all $j=\overline{1,d}$. Then we have
   \begin{equation} \label{22}
      \|T_{\mathbf{N},d}(f,\mathbf x)\|_\infty \le K_{\boldsymbol \delta}(\mathbf N,\mathbf M)\cdot \|f\|_q\,
   \end{equation}
where
   \begin{align} \label{23}
      K_{\boldsymbol \delta}(\mathbf N,\mathbf M) &= \Big(\frac8{q \pi^2}\Big)^{d/q}\, \prod_{j=1}^d
                            \frac{(e^{q \pi \delta_j/2}-1)^{1/q}}{\delta_j^{2/q}}
                            \Bigg(\sum_{k=1}^d \Big\{C_1(N_k,M_k)
                            \mathop{\prod_{j=1}^d}_{j\neq k} \Big( C_1(N_j,M_j) \nonumber \\
                  &\qquad + C_3^p(N_j,M_j) + C_4(N_j)\,C_2^p(N_j,M_j,\delta_j)\Big)\Big\}\Bigg)^{1/p} \, ,
   \end{align}
and
   \begin{align*}
              C_1(N,M) &:= 2\,\left( \frac{2^{2M}(M+1/2)(1+2M)^{1-2M}}{(1-M)^M (2N+1-M)^{2N+1-M}}  \right.\\
                       &\qquad \times \left. \frac{(N-1/2)(2N+1)^{2N+1}(N+M)^{2(N+M)}}
                                                         {(N-M-1/2)N^{2N+1}}\right)^p\Big( 1 + \frac N{p-1}\Big); \\
       C_2(N,M,\delta) &:= C_3(N,M)(M+1/2)(1+M/\delta);\\
              C_3(N,M) &:= \frac{2^{2M}(1-M)^{2M}}{{\rm sinc}(M)\,(1+2M)^{2M-1}(1-2M)^{4M}}\, \cdot\\
                       &\qquad \times \frac{(N-1-M)^{2(N-1-M)}(N+M)^{2(N+M)}}
                           {(N-1-2M)^{2(N-1-2M)}N^{2N}}\, ;\\
                C_4(N) &:= \frac{1}{p-1}\left(2^{p-1}(2p-1)+p-(N-1/2)^{1-p}-(N-1)^{1-p}\right)\,.
   \end{align*}
\end{theorem}
\begin{proof}
As we mentioned earlier in section 3, the function $Y_{\mathfrak J_{\mathbf x}}(f;\mathbf{x})$ does not depend on samples in
$\mathfrak T\setminus \{t_\mathbf{n}:\,\mathbf{n}\in \mathfrak J_{\mathbf x}\}$ and we assume that
outside $\mathfrak J_{\mathbf x}$ it is $\mathfrak h \equiv 0.$  Therefore, the structure of
$\{t_{n_j}:\,n_j\not\in \mathfrak J_{x_j}\}$ becomes uniform $t_{n_j} \equiv n_j$. In this case Theorem~\ref{th_ms}
guarantees that $f(\mathbf x)$ admits the representation (\ref{21}), and the so evaluated model (\ref{4}) gives us
   \[  \big| T_{\mathbf{N},d}(f;\mathbf x)\big|  \le
                          \underbrace{
                          \Big(\sum_{\mathbf n \in \mathbb Z^d \setminus \mathfrak J_{\mathbf x}}
                          \prod_{j=1}^d \Bigl|\frac{G_{N_j}(x_j,x_j)}{G_{N_j}'({x_j,t_{n_j}})
                          (x_j-{t_{n_j}})}\Bigr|^p\Big)^{1/p}}_{A_p}\,
                          \underbrace{
                          \Big(\sum_{\mathbf n \in \mathbb Z^d \setminus \mathfrak J_{\mathbf x}}|f(\mathbf{n})|^q
                          \Big)^{1/q}}_{B_q}\,.   \]
The multiplicative structure of $S(\mathbf{x},t_\mathbf{n})$ enables to estimate $A_p$ in the following way
   \begin{align} \label{26}
      A_p^p &\le \sum_{k=1}^d\sum_{n_k\in \mathbb Z\setminus  \mathfrak J_{x_k}}
                 \Big|\frac{G_{N_k}(x_k,x_k)}{G_{N_k}'({x_k,n_k})(x_k-{n_k})}\Big|^p
                 \mathop{\prod_{j=1}^d}_{j \neq k}\, \sum_{n_j \in \mathbb Z}
                 \Big|\frac{G_{N_j}(x_j,x_j)}{G_{N_j}'(x_j,t_{n_j})(x_j-t_{n_j})}\Big|^p \nonumber \\
            &=   \sum_{k=1}^d\sum_{n_k\in \mathbb Z\setminus \mathfrak J_{x_k}}
                 \Big|\frac{G_{N_k}(x_k,x_k)}{G_{N_k}'({x_k,n_k})(x_k-{n_k})}\Big|^p
                 \mathop{\prod_{j=1}^d}_{j\neq k}\,\Bigg(\sum_{n_j \in \mathbb Z\setminus \mathfrak J_{x_j}}
                 \Big|\frac{G_{N_j}(x_j,x_j)}{G_{N_j}'({x_j,n_j})(x_j-n_j)}\Big|^p \nonumber \\
            &\qquad + \sum_{n_j \in \mathfrak J_{x_j}}
                 \Big|\frac{G_{N_j}(x_j,x_j)}{G_{N_j}'(x_j,t_{n_j})(x_j-t_{n_j})}\Big|^p\Bigg)\, .
   \end{align}
Let us estimate $\sum_{n \in \mathbb Z\setminus \mathfrak J_{x}}\Big|\frac{G_{N}(x,x)}{G_{N}'(x,{n})(x-{n})}\Big|^p.$
Note, that due to our assumptions
   \[|\psi_{N}(n,x)|:= \left|\frac{G_{N}(x,x)}{G_{N}'(x,{t_n})(x-t_n)}\right| =
                       \left|\frac{\sin(\pi x)}{\pi(x-n)}\prod_{j\in\mathfrak J_{x}}
                       \frac{(t_j-x)(j-n)}{(t_j-n)(j-x)}\right|  \, , \]
for all $n \in \mathbb Z\setminus \mathfrak J_{x}$. Hence
   \[|\psi_{N}(n,x)|= \Bigg|\,{\rm sinc}(x-j_x)\frac{(t_{j_x}-x)(j_x-n)}{(x-n)(t_{j_x}-n)}
                      \mathop{\prod_{|j-x|\le N}}_{j \neq j_x} \frac{(t_j-x)(j-n)}{(t_j-n)(j-x)}\Bigg|\, , \]
where $j_x$ denotes the index closest to $x$, i.e. $j_x-1/2\le x<j_x+1/2\,.$

Due to $|h_{j_x}|\le M$ we have
   \begin{equation} \label{est1}
      \Big|{\rm sinc}(x-j_x)\,\frac{t_{j_x}-x}{x-n}\Big|\le \frac{M+1/2}{|x-n|}\,;
   \end{equation}
and
   \begin{align} \label{est2}
      \Big|\frac{j_x-n}{t_{j_x}-n}\Big| &\le 1+\frac{M}{|x-n|-M-1/2}\le 1+\frac{M}{N-M-1/2}\,;\\
      \Big|\frac{j-n}{t_j-n}\,\frac{t_j-x}{j-x}\Big| &\le \Big(1+\frac{M}{|j-n|-M}\Big)\cdot
                                                      \Big(1+\frac{M}{|j-x|}\Big)\, . \nonumber
   \end{align}
Then, being
   \[ \mathop{\prod_{|j-x|\le N}}_{j\neq j_x}\, \Big|\frac{(t_j-x)(j-n)}{(t_j-n)(j-x)}\Big| \le
              e^{\displaystyle \mathop{\sum_{|j-x| \le N}}_{j \neq j_x}\,\Big\{ \ln\Big(1+\frac{M}{|j-n|-M}\Big)+
                                 \ln \Big(1+\frac{M}{|j-x|}\Big) \Big\} }\]
let us estimate the first sum:
   \begin{align*}
      \mathop{\sum_{|j-x|\le N}}_{j\neq j_x} \ln  \Big(1 + \frac{M}{|j-n|-M}\Big)
              & < \sum_{k=1}^{2N+1}\ln\Big(1 +\frac{M}{k-M}\Big)
            \le \int_1^{2N+1}\ln\Big(\frac{t}{t-M}\Big)dt \\
              - \ln\left(1-M\right) &= \ln\frac{(2N+1)^{2N+1}}{(2N+1-M)^{(2N+1-M)}(1-M)^M}\, ,
   \end{align*}
and the second one as
   \begin{align} \label{est5}
      \mathop{\sum_{|j-x|\le N}}_{j\neq j_x} \ln\Big(1&+\frac{M}{|j-x|}\Big)
               \le 2\Big\{ \ln(1+2M)+\int_{1/2}^N\ln\Big(1+\frac{M}{t}\Big)\,dt\Big\} \nonumber \\
       &\qquad \le  \ln\,2^{2M}(1+2M)^{1-2M}(N+M)^{2(N+M)}N^{-2N}\, .
                 \end{align}
Collecting all estimates (\ref{est1})-(\ref{est5}), we conclude
   \begin{align*}
      \big|\psi_{N}(n,x)\big| &\le \frac{M+1/2}{|x-n|}\Big(1+\frac{M}{N-M-1/2}\Big)\,
                                   \frac{2^{2M}(1+2M)^{1-2M}}{(1-M)^M} \\
                         &\qquad \times \frac{(2N+1)^{2N+1}(N+M)^{2(N+M)}}{(2N+1-M)^{2N+1-M}N^{2N}} \, ,
   \end{align*}
and hence
   \begin{align*}
      &\sum_{n \in \mathbb Z\setminus  \mathfrak J_{x}}\Big|\frac{G_{N}(x,x)}{G_{N}'({x,n})(x-{n})}\Big|^p
             = \sum_{n \in \mathbb Z\setminus  \mathfrak J_{x}}\big|\psi_{N}(n,x)\big|^p
             \le \sum_{n \in \mathbb Z\setminus  \mathfrak J_{x}}\frac1{|x-n|^p} \\
            &\,\, \times \Bigg\{\frac{2^{2M}(N-1/2)(M+1/2)(2N+1)^{2N+1}(1+2M)^{1-2M}(N+M)^{2(N+M)}}
            {(N-M-1/2)(2N+1-M)^{2N+1-M}(1-M)^M N^{2N}}\Bigg\}^p\, .
   \end{align*}
Thus, we proceed evaluating
   \[ \sum_{n \in \mathbb Z\setminus  \mathfrak J_{x}}\frac{1}{|x-n|^p}\le
        2\left(\frac{1}{N^{p}}+\int_N^\infty\frac{dt}{t^p}\right) =
        2\left(\frac{1}{N^{p}}+\frac{1}{(p-1)\,N^{p-1}}\right) \, ,\]
such that gives
   \[ \sum_{n \in \mathbb Z\setminus  \mathfrak J_{x}} \Big|\frac{G_{N}(x,x)}{G_{N}'({x,n})(x-{n})}\Big|^p
      \le  C_1(N,M).\]
Now, let us evaluate $\sum_{n\in \mathfrak J_{x}}|\psi_N(n,x)|^p$, the second addend in \eqref{26}. Note,
that due to the declared assumptions of the Theorem, we have
   \[ |\psi_{N}(n,x)| = \Bigg|\frac{(n-t_n)\sin(\pi x)}{(n-x)\sin(\pi t_n)}
                         \mathop{\prod_{j\in \mathfrak J_{x}}}_{j\neq n}
                         \frac{(t_j-x)(j-t_n)}{(t_j-t_n)(j-x)}\Bigg| \qquad \big( n \in \mathfrak J_{x}\big)\,.\]
Let us consider the case $n \neq = j_x$. Then
   \[ \big|\psi_{N}(n,x)\big| = \Bigg|\frac{(n-t_n)\sin(\pi x)}{(n-x)\sin(\pi t_n)}
                                \frac{(t_{j_x}-x)(j_x-t_n)}{(j_x-x)(t_{j_x}-t_n)}
                                \mathop{\prod_{|j-x|\le N}}_{j\neq j_x,\,j\neq n}
                                \frac{(t_j-x)(j-t_n)}{(t_j-t_n)(j-x)}\Bigg|\,.\]
As {\rm sinc}$(t)$ decreases in $t\in [0,1/2]$, we get
   \begin{equation} \label{est1'}
      \Big|\frac{\pi(n-t_n)}{\sin(\pi t_n)}\Big|\le \frac1{{\rm sinc}(M)}\,,
   \end{equation}
and the inequality (\ref{est1}) is valid in our case too. Together with further necessary estimates, such as
   \begin{align} \label{est2'}
      \Big|\frac{j_x-t_n}{t_{j_x}-t_n}\Big| &\le 1+\frac{M}{\delta}\,; \nonumber\\
      \Big|\frac{j-t_n}{t_j-t_n} \cdot \frac{t_j-x}{j-x}\Big| &\le \Big(1+\frac{M}{|j-n|-2M}\Big)
                                       \Big(1+\frac{M}{|j-x|}\Big)\, ,
   \end{align}
we deduce
   \begin{align} \label{est3'}
      \mathop{\prod_{|j-x|\le N}}_{j\neq j_x,\,j\neq n} \Big|\frac{(t_j-x)(j-t_n)}{(t_j-t_n)(j-x)}\Big| &\le
           \exp\Big( \displaystyle\mathop{\sum_{|j-x|\le N}}_{j\neq j_x\,j\neq n}
           \Big\{ \ln\Big(1+\frac{M}{|j-n|-2M}\Big)\nonumber \\
           & \qquad + \ln\Big(1+\frac{M}{|j-x|}\Big)\Big\}\Big) \,.
   \end{align}
Let us estimate the first sum as
   \begin{align} \label{est4'}
      &\mathop{\sum_{|j-x|\le N}}_{j\neq j_x,\,j\neq n} \ln\left(1+\frac{M}{|j-n|-2M}\right)
         \le 2\sum_{k=1}^{N-1}\ln\left(1+\frac{M}{k-2M}\right)
         \le 2\left(\ln\left(1+\frac{M}{1-2M}\right)\right. \nonumber \\
      &\qquad \quad +  \left. \int_1^{N-1}\ln\left(1+\frac{M}{t-2M}\right)dt\right) =
         \ln \frac{(N-1-M)^{2(N-1-M)}(1-M)^{2M}}{(N-1-2M)^{2(N-1-2M)}(1-2M)^{4M}}\, .
   \end{align}
Making use of \eqref{est5} to the second sum in \eqref{est3'} we have
   \begin{equation} \label{est5'}
      \mathop{\sum_{|j-x|\le N}}_{j\neq j_x\,j \neq n} \ln \Big(1+\frac{M}{|j-x|}\Big)
         \le \ln \frac{2^{2M}(1+2M)^{1-2M}(N+M)^{2(N+M)}}{N^{2N}}\, .
   \end{equation}
Combining all estimates (\ref{est1'})-(\ref{est5'}), we obtain
   \[ \big|\psi_{N}(n,x)\big| \le \frac{C_2(N,M,\delta)}{|x-n|} \qquad
      \big(n \in \mathfrak J_{x}\setminus \{j_x\}\big)\,.\]
Let us consider the case $n=j_x.$ As sinc$(t)$ decreases, by (\ref{est1'}) it follows
   \[ \left|\frac{(n-t_n)\sin(\pi x)}{(n-x)\sin(\pi t_n)}\right|\le {\rm sinc}^{-1}(M)\,.\]
Hence, for $n=j_x$  we get
   \[|\psi_{N}(n,x)|\le \frac{C_2(N,M,\delta)}{(M+1/2)(1+M/\delta)} = C_3(N,M), \]
because (\ref{est3'})-(\ref{est5'}) are still valid in this case as well. Combining all upper bounds concerning
$|\psi_{N}(n,x)|$ we have
   \[ \sum_{n\in \mathfrak J_{x}}\Big|\frac{G_{N}(x,x)}{G_{N}'(x,t_{n})(x-t_{n})}\Big|^p\le
            C^p_3(N,M)+C_2^p(N,M,\delta)\mathop{\sum_{n \in \mathfrak J_{x}}}_{n\neq j_x} |x-n|^{-p}\,.\]
It remains to realize the estimate
   \begin{align*}
      \mathop{\sum_{n \in \mathfrak J_x}}_{n \neq j_x} \frac{1}{|x-n|^p} &\le
              \sum_{n=1}^N \frac{1}{(n-1/2)^p}+\sum_{n=1}^{N-1}\frac{1}{n^p} \\
        & \le 2^p+\int_1^N\frac{dt}{(t-1/2)^p}+1+\int_1^{N-1}\frac{dt}{t^p} = C_4(N)\, ,
   \end{align*}
such that gives
   \[\sum_{n\in \mathfrak J_{x}}\Bigl|\frac{G_{N}(x,x)}{G_{N}'(x,t_{n})(x-t_{n})}\Bigr|^p\le
           C^p_3(N,M)+C_4(N)\,C_2^p(N,M,\delta). \]
Therefore, by (\ref{26})
   \begin{align} \label{Ap}
      A_p^p &\le \sum_{k=1}^d \Bigg(C_1(N_k,M_k)
                 \mathop{\prod_{j=1}^d}_{j\neq k}\Big( C_1(N_j,M_j) \nonumber \\
               & \qquad + C^p_3(N_j,M_j) +
                 C_4(N_j)\,C_2^p(N_j,M_j,\delta_j) \Big)\Bigg)\,.
   \end{align}
To estimate $B_q$ we use Theorem~\ref{th1} with $\sigma_1= \cdots = \sigma_d= \pi$. This results in
   \begin{equation} \label{B_destimate}
      B_q^q\le \mathfrak B_{d,q} = \left(\frac{8}{q \pi^2}\right)^d\, \prod_{j=1}^d \frac{e^{q \pi
                                   \delta_j/2}-1}{\delta_j^2}.
   \end{equation}
Now, collecting all these involved estimates, we arrive at (\ref{22}) and (\ref{23}).
\end{proof}

Denote $\underline{\delta}:=\min_{j=\overline{1,d}}{\delta_j}$ and
$\overline{\delta}:=\max_{j=\overline{1,d}}{\delta_j}.$

Let us consider the Paley-Wiener space
$PW_{\pi,d}^q,\ q \ge 1$ of all complex - valued $L^q(\mathbb
R^d)$--functions whose Fourier spectrum is bandlimited to
$[-\pi,\pi]^d$.

\begin{corollary}\label{cor} Let $f\in PW^q_{\pi,d},$ $q\ge 1,$ $\widetilde{M}$ satisfy {\rm(\ref{M}).} Then we have
   \[ \| T_{\mathbf{N},d}(f,\mathbf x)\|_\infty \le \widetilde{K}(\mathbf N, \widetilde{M},
      \underline{\delta},\overline{\delta}) \cdot\|f\|_q\,\]
where
   \begin{align}
      \widetilde{K}(\mathbf N, &  \widetilde{M},\underline{\delta},\overline{\delta})
                                = \Bigg(\sum_{k=1}^d \Big(C_1(N_k,\widetilde{M})
                                  \prod_{\stackrel{j=\overline{1,d}}{j\not=k}}\Big(
                                  C_1(N_j,\widetilde{M})+C^p_3(N_j,\widetilde{M})+C_4\,(N_j)\nonumber\\
                          &\times C_2^p(N_j,\widetilde{M},\underline{\delta})\Big)\Big)\Bigg)^{1/p}
                                  \max{}^{d/q}\left( \frac{8(e^{q \pi
                                  \underline{\delta}/2}-1)}{q \pi^2\underline{\delta}^2},\ \frac{8(e^{q \pi
                                  \overline{\delta}/2}-1)}{q \pi^2\overline{\delta}^2}\right)\, .\label{K}
   \end{align}
\end{corollary}
\noindent {\bf Proof:} We can use the estimate (\ref{Ap}) for $A_p$ from Theorems~\ref{th_ir} with $\widetilde{M}$ instead
of all $M_j$   and $\underline{\delta}$ instead of all $\delta_j$  (in
this case $\mathfrak T$ could contain some additional $t_\mathbf{n}$
which might results only in increasing $A_p$).

From
 \[\left(\frac{e^{q \pi \delta_j/2}-1}{\delta_j^2}\right)'=\frac{ e^{q \pi\delta_j/2}(q\pi\delta_j-4)+4 }
              {2\delta_j^3}\]
follows that $(e^{q \pi \delta_j/2}-1)/\delta_j^{2}$ decreases when $\delta_j\in[0,\delta^*)$ and increases
for $\delta_j>\delta^*$, where $\delta^*:=\frac{2z^*}{q\pi}$,
   \[z^*:={\rm LambertW}(-2/e^2)+2\approx 1.59362.\]
To find the numerical value of $z^*$ we use the {\it Mathematica's} in--built numerical routine {\tt ProductLog[z]}
for the computation of ${\rm LambertW}$ function which is the inverse function of $ze^z$, \cite{Wei}. Namely,
we use that branch of ${\rm LambertW}$ which gives the nonzero solution of the equation $(2-z)e^{z}=2$.
Thus, $\delta^*=\frac{2({\rm LambertW}(-2/e^2)+2)}{q\pi}\approx
\frac{1.29174}{q}$. Therefore by (\ref{B_destimate}) we have the estimate
   \[{\mathfrak B_d}^{1/d} \le \max\left( \frac{8(e^{q \pi
      \underline{\delta}/2}-1)}{q \pi^2\underline{\delta}^2},\ \frac{8(e^{q \pi
      \overline{\delta}/2}-1)}{q \pi^2\overline{\delta}^2}\right).  \qquad \qquad \Box \]
\begin{remark}
For $q<\frac{2z^*}{\pi}$ it holds $\max\left( \frac{8(e^{q \pi
      \underline{\delta}/2}-1)}{q \pi^2\underline{\delta}^2},\ \frac{8(e^{q \pi
      \overline{\delta}/2}-1)}{q \pi^2\overline{\delta}^2}\right)=\frac{8(e^{q \pi
      \underline{\delta}/2}-1)}{q \pi^2\underline{\delta}^2}.$
\end{remark}
\begin{corollary} Let $f\in PW^q_{\pi,d},$ $q\ge 1,$ $d\in \mathbb N,$
   \begin{align}
      \widetilde{M} &\le \frac{1}{4}, \qquad \qquad \qquad \qquad \qquad \mbox{if}\ \ q=1, \ d=1; \nonumber\\
      \widetilde{M} &<  \min\left(\frac{1}{4q},\frac{1}{(4d-1)q}\right),\,\,\,\quad \mbox{else},\label{M1}
   \end{align}
and $\widetilde{N}\to +\infty$ in such way that ${\displaystyle\max_{k,j=\overline{1,d}}}\ {N_j}/{N_k}= O(1)$.
Then it holds
   \[ \| T_{\mathbf{N},d}(f;\mathbf x)\|_\infty \to 0. \]
\end{corollary}
\begin{proof} By definitions for constants $C_j,\ j=\overline{1,4}$ from Theorem~\ref{th_ir} we get
   \begin{align*}
      C_1(N,M) &\sim const\cdot N^{1+3Mp-p},\quad C_2(N,M,\delta)\sim const\cdot N^{4M},\\
      C_3(N,M) &\sim const\cdot N^{4M},\quad C_4(N)\sim const\,,
   \end{align*}
when $N\to+\infty.$

Therefore by (\ref{K}) $\widetilde{K}(\mathbf N, \widetilde{M},\widetilde{\delta})\to 0$ when $\widetilde{N}\to \infty$
if$\widetilde{M}<\frac{1}{(4d-1)q}.$ Application of Corollary~\ref{cor} and (\ref{M}) finishes proof and gives us
conditions (\ref{M1}).
\end{proof}

\section{Final remarks}
In most papers known in literature (see, for example, \cite{hig},
\cite{li}, \cite{LF}, and references in them) to obtain
Kotelnikov-Shannons theorems and upper bounds for approximation
errors in $L^p(\mathbb R)$ spaces various authors considered
particular classes of functions with prescribed decay conditions
which gave opportunity to estimated $B_q$ in (\ref{4}) in such
manner that the estimate vanishes when $N$ tends to infinity. The
another term $A_p$ in (\ref{4}) usually was estimated by some
constant which did not depend on $N.$

In this paper we propose estimates for $A_p$ which depends on $N$
and vanishes when $N$ tends to infinity. This approach gives
opportunity to consider and obtain approximation errors estimates
for wide functional classes without strong assumptions on decay
behaviour.  New upper error bounds in $\|\cdot\|_\infty$ norm in
sampling theorem were obtained. The case of multidimensional
irregular sampling was considered.

Results presented in the paper and numerical simulation done by the
authors (consult \cite{uni_1}) arise some new open problems:
\begin{enumerate}
\item To obtain sharp estimates in Theorem \ref{th1};
\item To obtain sharp estimates in Theorem \ref{th_ir} (for uniform sampling and $p=2$ such sharp
estimates were derived in \cite{OP2});
\item To apply obtained
results to stochastic case, see \cite{OP}.
\end{enumerate}
In the case (3) we point out two approaches. The first one concerns the so--called spectral representation
$X(t) = \int_\Lambda f(t,\lambda)Z(d\lambda)$ of given stochastic process
$X(t),\,  t\in T\subseteq \mathbb R$, with deterministic kernel function $f(t,\lambda)$ coming from
function space with respect to the first argument, e.g. Paley--Wiener class $PW_{\pi, d}^2$ \cite{OP1},
\cite{OP2}, \cite{Pog3}, Bernstein class $B_{\sigma, p}$ \cite{OP4}. Under suitable conditions the
kernel function possesses WKS sampling sum expansion in the related $L_p$--norm.  Again,
by the quoted spectral representation formula one deduces the related WKS sampling restoration sum for the initial
stochastic signal $X(t)$ in the mean--square and/or almost sure sense. Second, when $X(t)$ weak sense stationary,
it is enough to consider the WKS sampling
restoration sum for the covariance function $R_X(\tau)=\mathsf EX(\tau)X(0), \tau \in \mathbb R$ of the
considered stochastic signal. Then we deduce the final mean square WKS sampling restoration sum for $X(t)$.
In both cases the main
mathematical tool will be the celebrated Karhunen--Cram\'er theorem on integral
representation of stochastic processes \cite[pp. 144--179, p. 156]{Yag}  However, the structure
of considered stochastic signals should mainly influence the convergence conditions and rates of
derived WKS sampling sums.

\section*{Acknowledgement}The authors were supported by the La Trobe University Research
Grant-501821 "Sampling, wavelets and optimal stochastic modelling".

\end{document}